\documentclass[format=acmsmall, review=false]{acmart}
\usepackage{acm-ec-17}
\usepackage{booktabs} % For formal tables

\begin{document}
% Title portion. Note the short title for running heads 
\title[Simple Pricing Schemes for the Cloud]{Simple Pricing Schemes for the Cloud}  
\author{Ian A. Kash}
\affiliation{%
  \institution{Microsoft Research}
  \city{Cambridge}
  \country{UK}}
\author{Peter Key}
\affiliation{%
  \institution{Microsoft Research}
  \city{Cambridge}
  \country{UK}}
\author{Warut Suksompong}
\affiliation{%
  \institution{Stanford University}
  \city{Stanford}
  \country{USA}}

% note that the abstract must come before \maketitle
\begin{abstract}
The problem of pricing the cloud has attracted much recent attention due to the widespread use of cloud computing and cloud services. From a theoretical perspective, several mechanisms that provide strong efficiency or fairness guarantees and desirable incentive properties have been designed. However, these mechanisms often rely on a rigid model, with several parameters needing to be precisely known in order for the guarantees to hold. In this paper, we consider a stochastic model and show that it is possible to obtain good welfare and revenue guarantees with simple mechanisms that do not make use of the information on some of these parameters. In particular, we prove that a mechanism that sets the same price per time step for jobs of any length achieves at least $50\%$ of the welfare and revenue obtained by a mechanism that can set different prices for jobs of different lengths, and the ratio can be improved if we have more specific knowledge of some parameters. Similarly, a mechanism that sets the same price for all servers even though the servers may receive different kinds of jobs can provide a reasonable welfare and revenue approximation compared to a mechanism that is allowed to set different prices for different servers.
\end{abstract}

\maketitle

\section{Introduction}

With cloud computing generating billions of dollars per year and forming a significant portion of the revenue of large software companies \cite{Columbus16}, the problem of how to price cloud resources and services is of great importance. On the one hand, for a pricing scheme to be used, it is necessary that the scheme provide strong welfare and revenue guarantees. On the other hand, it is also often desirable that the scheme be simple. We combine the two objectives in this paper and show that simple pricing schemes perform almost as well as more complex ones with respect to welfare and revenue guarantees. In particular, consider the pricing scheme for virtual machines on Microsoft Azure shown in Figure \ref{fig:azure}. Once the user chooses the basic parameters such as region, type, and instance size, the price is calculated by simply multiplying an hourly base price by the number of virtual machines and number of hours desired. The question that we study can be phrased in this setting as follows: How much more welfare or revenue could be created if instead of this simple multiplication formula, a complex table specifying the price for each number of hours were to be used?  Our main result is that the former offers at worst a two approximation to the latter, both in terms of welfare and revenue.  Similarly, we demonstrate that setting a single price for a group of servers, even though the servers may receive different kinds of jobs, can provide a reasonable welfare and revenue approximation compared to setting different prices for different servers.

\begin{figure}[!ht]
\centering
\includegraphics[width=0.8\textwidth]{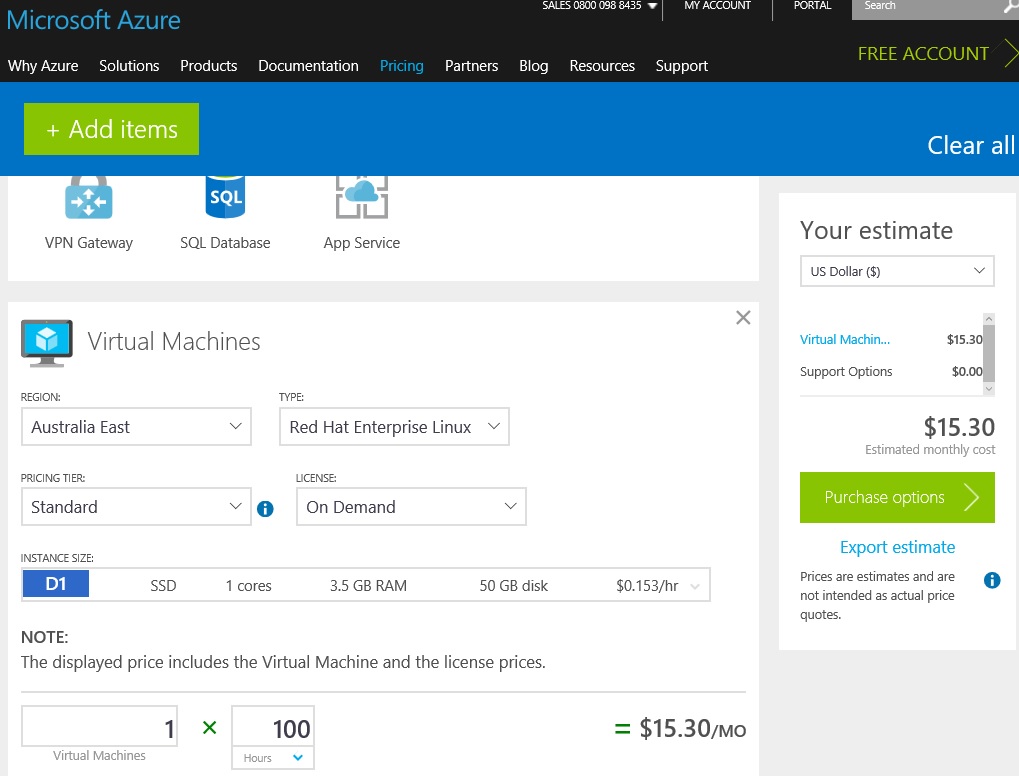}
\caption{Pricing scheme for virtual machines on  Microsoft Azure \cite{Azure16}.}
\label{fig:azure}
\end{figure}

In much of the prior work in this space, which focuses more explicitly on scheduling, prices depend in a complex way on a number of parameters (typically including job length, arrival time, deadline, and value) as well as the current state of the system \cite{AzarKaLu15,DehghaniKaKe16,JainMeNa11,JainMeNa12,LucierMeNa13}.  A weakness of such schemes is that they require these parameters to be known up front in order for the desirable properties of the mechanisms, such as their approximation ratios, to hold. The availability of such information is not always realistic in practice.  Even when it is in principle possible to provide this information, there is a cost to participants in both time and resources to figure it out. In this work, we show that good results are possible with no up front information.

For our initial results we assume that there is a single server, which receives jobs of various lengths whose value per time step is drawn from the same probability distribution regardless of length. We compare the welfare and revenue that can be obtained by setting a price per time step that is independent of the job length against the corresponding objective obtained by setting an individual price for each job length. When we are allowed the freedom of setting different prices for different job lengths, intuitively we want to set a higher price per time step for longer jobs as a premium for reserving the server for a longer period of time.\footnote{\label{footnote:amazon}Amazon recently started offering a product called ``defined duration spot instances'' where users can specify a duration in hourly increments up to six hours. Indeed, the price \emph{per hour} of this product increases as the number of hours increases. (See Figure~\ref{fig:amazon}.)} However, as we show, we do not lose more than $50\%$ of the welfare or revenue if we are only allowed to set one price. We would like to emphasize that this is a worst-case bound over a wide range of parameters, including the number of job lengths, the distribution over job lengths, and the distribution over job values. Indeed, as we show, we can obtain improved bounds if we know the value of some of these parameters. The price that we use in the single-price setting can be chosen from one of the prices used in the multi-price setting, meaning that we do not have to calculate a price from scratch. Moreover, all of our approximation guarantees hold generally for arbitrary prices, meaning that for any prices that we may set in a multi-price setting (i.e., not necessarily optimal ones), we can obtain an approximation of the welfare or revenue by setting one of those prices alone. Finally, we emphasize that these results put no restrictions on the form of the distribution; it can be discrete, continuous, or mixed.  The only substantive constraint is that jobs of all lengths share the same distribution of value per time step.  However, in an extension we show that a version of our results continues to hold even if this constraint is relaxed.

\begin{figure}[!ht]
\centering
\includegraphics[width=0.8\textwidth]{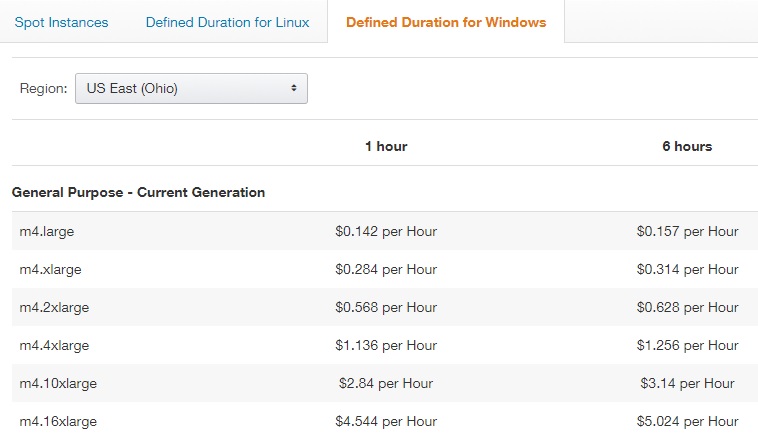}
\caption{Pricing scheme for defined duration spot instances on Amazon \cite{Amazon17}.}
\label{fig:amazon}
\end{figure}

We then generalize our results to a setting where there are multiple servers, each of which receives jobs of various lengths. The distribution over job lengths can be different for different servers. This is conceivable, for instance, if the servers are in various geographic locations or are utilized by various groups of users. We compare the welfare and revenue obtained by a simple pricing scheme that sets the same price for all servers against the corresponding objective achieved by a scheme that can set a different (single) price for each server. Roughly speaking, we show that as long as the parameters are not too extreme, e.g., the number of servers or the job lengths are not too large, then we do not lose too much of the welfare or revenue by setting a single price. Combining this with our initial results, we obtain an approximation of a very restricted pricing scheme where we must set the same price for all servers and all job lengths against one where we can set an individual price for each job length of each server. These results require an assumption that all servers have the same probability of not receiving a job at a time step. Using similar techniques, we also obtain approximation bounds when this assumption does not hold but there is only one job length across all servers.
%The analysis of the latter setting is presented in Appendix \ref{sec:multserversonelength}.

\subsection{Related work}

Much recent work has focused on designing online scheduling mechanisms with good welfare guarantees and incentive properties. \citet{JainMeNa11} exhibited a truthful mechanism for batch jobs on cloud systems where jobs are allocated non-preemptively, and the same group of authors came up with mechanisms for deadline-sensitive jobs in large computing clusters \cite{JainMeNa12}. \citet{LucierMeNa13} also considered the problem of scheduling deadline-sensitive jobs; they circumvented known lower bounds by assuming that jobs could be delayed and still finish by their deadline. \citet{ZhangLiJi13} developed a framework for truthful online cloud auctions where users with heterogeneous demands can come and leave on the fly. More recently, \citet{AzarKaLu15} constructed a truthful mechanism that achieves a constant competitive ratio given that slackness is allowed, while \citet{DehghaniKaKe16} assumed a stochastic model and developed a truthful mechanism that approximates the expected maximum welfare up to a constant factor.
\citet{WangMaQin15} designed mechanisms for selling reserved instances where users are allowed to reserve resources of any length and from any time point in the future. Their mechanisms determine the acceptance and payment immediately when the job arrives, and achieve a competitive ratio that is optimal to within a constant factor with regard to welfare.

Other work in this space has dealt with comparing pricing mechanisms such as the on-demand market and the spot market \cite{AbhishekKaKe12,DierksSe16,HoyImLu16}, achieving fairness in job allocation \cite{FriedmanGhPs14}, and studying models of real-time pricing with budget constraints \cite{FriedmanRaSh15}. \citet{KashKe16} gave a survey of the current state of research in economics and computer science with respect to cloud pricing.

From a technical perspective, our work bears a resemblance to the work of \citet{DuttingFiKl16} on discriminatory and anonymous posted pricing and of \citet{DisserFeGa16} on hiring secretaries. In particular, D\"{u}tting et al. considered the problem of selling a single item to buyers who arrive sequentially with values drawn independently from identical distributions. They showed that by posting discriminatory prices, one can obtain at most $2-1/n$ times as much revenue as that obtained by posting the same anonymous price, where $n$ is the number of buyers. As is also the case in our work, their anonymous price can always be chosen from one of the discriminatory prices, but their bound is obtained via a relaxation of the discriminatory pricing problem, a different technique than what we use. Disser et al. provided a competitive online algorithm for a variant of the stochastic secretary problem, where applicants need to be hired over time. Upon the arrival of each applicant, the cost per time step of the applicant is revealed, and we have to decide on the duration of the employment, starting immediately. Once an applicant is accepted, we cannot terminate the contract until the duration of the job is over.

Our work falls into the broader area of the design and analysis of simple mechanisms, particularly posted price mechanisms. One of the motivations for studying simple mechanisms is that in practice, designers are often willing to partially give up optimality in return for simplicity. Mechanisms that simply post prices on goods have received significant attention since they reflect perhaps the most common way of selling goods in the real world, and moreover they leave no room for strategizing, making them easy for agents to participate in. A long line of work has investigated how well such mechanisms can approximate optimal mechanisms with respect to various objectives including welfare \cite{FeldmanGrLu15,CohenaddadEdFe16,EzraFeRo17}, revenue \cite{ChawlaHaMa10,BabaioffBlDu11,BlumrosenHo08}, and social costs \cite{CohenEdFi15}.
In Section~\ref{subsec:extension} we show that techniques from this literature can recover some of our results under relaxed assumptions.

\section{Preliminaries}

We consider a system with a number of servers and discrete time steps. Each job takes an integer number of time steps to complete and yields a value upon completion. The value \emph{per time step} of a job is drawn from a distribution which is independent of the length of the job. Let $F$ be the cumulative distribution function of this distribution and $f$ the probability density function with respect to a base measure $\mu$, and define $\ell(x)=xf(x)$.\footnote{For technical reasons, we will deviate slightly from the usual notion of cumulative distribution function. In particular, if $y$ is a random variable drawn from a distribution, then we define its cumulative distribution function $F(x)$ as $\text{Pr}[y<x]$ instead of the usual $\text{Pr}[y\leq x]$. This will only be important when we deal with discrete distributions.} We do not make any assumption on our distribution; in particular, it need not be continuous or discrete, which is why we allow flexibility in terms of the base measure.

When a job request is made for a job to be served by a server, there is a price $p$ per time step which may depend on the job length and/or the server. If the value per time step of the job is at least $p$, the server accepts and executes the job to completion. Otherwise, the server rejects the job. The objectives in our model are the steady-state welfare and revenue for each pricing scheme. In particular, we will be interested in the expected welfare and revenue per time step, given that the job values are drawn from a probability distribution. This can also be thought of as the average welfare and revenue per time step that result from a pricing scheme over a long period of time.

In Section \ref{sec:oneserver}, we assume that there is a single server. Each time step, either zero or one job appears. A job with length $a_i$ appears with probability $0<r_i\leq 1$, where $\sum_{i=1}^n r_i\leq 1$ and $n$ denotes the number of job lengths. We are allowed to set a price $p_i$ for jobs of length $a_i$. If a server accepts a job of length $a_i$, it is busy and cannot accept other jobs for $a_i$ time steps, including the current one. We compare the setting where we are forced to set the same price $p$ for all job lengths against the setting where we can set a different price $p_i$ for each job length $a_i$. Note that if we could set different prices for different job lengths, then to optimize welfare or revenue, intuitively we would set a higher price per time step for longer jobs as a premium for reserving the server for a longer period. Put differently, once we accept a longer job, we are stuck with it for a longer period, during which we miss the opportunity to accept other jobs. Consequently, we should set a higher standard for accepting longer jobs. (See also Footnote~\ref{footnote:amazon}.)

In Section \ref{sec:multservers}, we assume that there are multiple servers. Each time step, either zero or one job appears for each server $1\leq j\leq n$. For server $j$, a job with length $a_{ji}$ appears with probability $0<r_{ji}\leq 1$ for $1\leq i\leq n_j$, where $n_j$ denotes the number of job lengths for server $j$. We do not assume that the set of job lengths or the number of job lengths is identical across servers. On the other hand, we assume that the probability of no job appearing at a time step is the same for all servers, i.e., $\sum_{i=1}^{n_j}r_{ji}$ is constant for any $j$. In Subsection \ref{subsec:multserversoneprice}, we assume that we can set one price per server, and we compare the setting where we are forced to set the same price $p$ for all servers against that where we can set a different price $p_j$ for each server $j$. In Subsection \ref{subsec:multserversmultprice}, we assume that we can set a different price $p_{ji}$ for each server $j$ and each of its job lengths $a_{ji}$, and we compare that setting against the one where we are forced to set the same price $p$ for all servers and all job lengths.

\section{One Server}
\label{sec:oneserver}

In this section, we assume that there is a single server, which receives jobs of various lengths. After presenting an introductory example in Subsection \ref{subsec:oneserverexample}, we consider the general setting with an arbitrary number of job lengths in Subsection \ref{subsec:oneservermanylengths}. In this setting, we show a $50\%$ approximation for both welfare and revenue of setting one price for all job lengths compared to setting an individual price for each job length, for any realization of the parameters. Moreover, we show in Subsection \ref{subsec:oneservertwolengths} that our techniques provide a template for deriving tighter bounds if we have more specific information on the parameters. In particular, when there are two job lengths, we show for each setting of the parameters a tight approximation bound for welfare and revenue. Our approximation results hold for arbitrary (i.e., not necessarily optimal) pricing schemes, and the price we use in the single-price setting can be drawn from one of the prices in the multi-price setting. Finally, in Subsection~\ref{subsec:extension} we consider an extension that does not assume independence between the job length and the value per time step.

\subsection{Warm-Up: Uniform Distribution}
\label{subsec:oneserverexample}

As a warm-up example, assume that at any time step a job with length 1 or 2 appears with probability $50\%$ each. The value per time step of a job is drawn from the uniform distribution over $[0,1]$. Suppose that we set a price per time step $p_1$ for jobs of length 1 and $p_2$ for jobs of length 2.

Consider an arbitrary time step when the server is free. If the job drawn at that time step has length 1, then with probability $p_1$ it has value below $p_1$ and is rejected. In this case, the server passes one time step without a job. Otherwise, the job has value at least $p_1$ and is accepted. In this case, the expected welfare from executing the job is $\frac{1+p_1}{2}$. Similarly, if the job has length 2, then with probability $p_2$ it is rejected, and with probability $1-p_2$ it is accepted and yields expected welfare $2\cdot\frac{1+p_2}{2}=1+p_2$ over two time steps. Letting $c_w$ denote the expected welfare per time step assuming that the server is free at the current time step, we have
$$0=\frac12\left(-p_1c_w+(1-p_1)\left(\frac{1+p_1}{2}-c_w\right)\right)+\frac12\left(-p_2c_w+(1-p_2)\left(1+p_2-2c_w\right)\right).$$
The two terms on the right hand side correspond to jobs of length 1 and 2, which are drawn with probability $1/2$ each. In the case that a job of length 2 is drawn, with probability $p_2$ it is rejected and the server is idle for one time step, during which it would otherwise have produced expected welfare $c_w$. With the remaining probability $1-p_2$ the job is accepted, yielding expected welfare $1+p_2$ over two time steps, during which the server would otherwise have produced expected welfare $2c_w$. The derivation for the term corresponding to jobs of length 1 is similar. By equating the expected welfare with the variable denoting this quantity, we arrive at the equation above.

Solving for $c_w$, we get
$$c_w(p_1,p_2)=\frac{\frac{(1-p_1)(1+p_1)}{2}+(1-p_2)(1+p_2)}{3-p_2}.$$

To maximize $c_w(p_1,p_2)$ over all values of $p_1,p_2$, we should set $p_1=0$. (Indeed, to maximize welfare we should always accept jobs of length 1 since they do not interfere with future jobs.) Then the value of $p_2$ that maximizes $c_w(p_1,p_2)$ is $p_2=3-\sqrt{\frac{15}{2}}\approx 0.261$, yielding $c_w(p_1,p_2)=6-\sqrt{30}\approx 0.522$.

On the other hand, if we set the same price $p=p_1=p_2$ for jobs with different lengths, our welfare per time step becomes
$$c_w(p)=\frac{\frac{(1-p)(1+p)}{2}+(1-p)(1+p)}{3-p}=\frac{3(1-p)(1+p)}{2(3-p)}.$$
This is maximized at $p=3-2\sqrt{2}\approx 0.172$, yielding $c_w(p)=9-6\sqrt{2}\approx 0.515$. Moreover, if we use either of the prices in the optimal price combination for the two-price setting as the single price, we get $c_w(0)=0.5$ and $c_w\left(3-\sqrt{\frac{15}{2}}\right)\approx 0.510$.

Next, we repeat the same exercise for revenue. We can derive the equations in the same way, with the only difference being that the revenue from accepting a job at price $p$ is simply $p$. Letting $c_r$ denote the revenue per time step, we have
$$0=\frac12\left(-p_1c_r+(1-p_1)\left(p_1-c_r\right)\right)+\frac12\left(-p_2c_r+(1-p_2)\left(2p_2-2c_r\right)\right).$$
Solving for $c_r$, we get
$$c_r(p_1,p_2)=\frac{(1-p_1)p_1+2(1-p_2)p_2}{3-p_2}.$$

To maximize $c_r$ over all values of $p_1,p_2$, we should set $p_1=0.5$. (Indeed, to maximize revenue we should always set the monopoly price for jobs of length 1 since they do not interfere with future jobs.) Then the value of $p_2$ that maximizes $c_r(p_1,p_2)$ is $p_2=3-\sqrt{\frac{47}{8}}\approx 0.576$, yielding $c_r(p_1,p_2)=10-\sqrt{94}\approx 0.304$.

On the other hand, if we set the same price $p=p_1=p_2$ for jobs with different lengths, our revenue per time step becomes
$$c_r(p)=\frac{(1-p)p+2(1-p)p}{3-p}=\frac{3(1-p)p}{3-p}.$$
This is maximized at $p=3-\sqrt{6}\approx 0.551$, yielding $c_r(p)=15-6\sqrt{6}\approx 0.303$. Moreover, if we use either of the prices in the optimal price combination for the two-price setting as the single price, we get $c_r(0.5)=0.3$ and $c_r\left(3-\sqrt{\frac{47}{8}}\right)\approx 0.302$.

Observe that for both welfare and revenue, the maximum in the one-price setting is not far from that in the two-price setting. In addition, in both cases at least one of the two prices in the optimal price combination for the two-price setting, when used alone as a single price, performs almost as well as the maximum in the two-price setting. In the remainder of this section, we will show that this is not a coincidence, but rather a phenomenon that occurs for any set of job lengths, any probability distribution over job lengths, and any probability distribution over job values.

\subsection{General 50\% Approximation}
\label{subsec:oneservermanylengths}

In this subsection, we consider a general setting with an arbitrary number of job lengths. We show that even at this level of generality, it is always possible to obtain $50\%$ of the welfare and revenue of setting an individual price for each job length by setting just one price. Although the optimal price in the one-price setting might be different from any of the prices in the multi-price setting, we show that at least one of the prices in the latter setting can be used alone to achieve the $50\%$ guarantee. 

Assume that there are jobs of lengths $a_1\le a_2\le\dots\le a_n$ which appear at each time step with probability $r_1,r_2,\dots,r_n$, respectively. Suppose that we set a price per time step $p_i$ for jobs of length $a_i$. Recall that the value per time step of a job is drawn from a distribution with cumulative distribution function $F$ and probability density function $f$.

The following lemma gives the formulas for the expected welfare and revenue per time step. The derivation is straightforward and can be found in the Appendix \ref{app:prooflemmaformulas}.

\begin{lemma}
\label{LEM:FORMULAS}
Let $S=a_1r_1+\dots+a_nr_n$ and $R=r_1+\dots+r_n$, and let $c_w$ and $c_r$ denote the expected welfare and revenue per time step, respectively. We have
\begin{equation}
\label{eqn:cw}
c_w(p_1,p_2,\dots,p_n)=\frac{a_1r_1\int_{x\geq p_1}\ell d\mu+\dots+a_nr_n\int_{x\geq p_n}\ell d\mu}{S-((a_1-1)r_1F(p_1)+\dots+(a_n-1)r_nF(p_n))+(1-R)}
\end{equation}
and
\begin{equation}
c_r(p_1,p_2\dots,p_n)=\frac{a_1r_1(1-F(p_1))p_1+\dots+a_nr_n(1-F(p_n))p_n}{S-((a_1-1)r_1F(p_1)+\dots+(a_n-1)r_nF(p_n))+(1-R)}.
\end{equation}

In particular, if $p_1=\dots=p_n=p$, then
$$c_w(p)=\frac{S\int_{x\geq p}\ell d\mu}{S-(S-R)F(p)+(1-R)}$$
and
$$c_r(p)=\frac{S(1-F(p))p}{S-(S-R)F(p)+(1-R)}.$$
\end{lemma}

With the formulas for welfare and revenue in hand, we are ready to show the main result of this section, which exhibits that the worst-case approximation ratio for welfare or revenue between the single-price setting and the multi-price setting is at least $50\%$. As we will see later in Subsection \ref{subsec:oneservertwolengths}, this bound is in fact tight, and it remains tight even when there are only two job lengths. Note that the bound holds for any number of job lengths, any distribution over job lengths, and any distribution over job values. 

\begin{theorem}
\label{THM:50PERCENTGENERAL}
For any prices $p_1,p_2,\dots,p_n$ that we set in the multi-price setting, we can achieve a welfare (resp. revenue, or any convex combination of welfare and revenue) approximation of at least $50\%$ in the one-price setting by using one of the prices $p_i$ as the single price.
\end{theorem}

To prove Theorem \ref{THM:50PERCENTGENERAL}, we work with the ratio $$\frac{\max(c_w(p_1),\dots,c_w(p_n))}{c_w(p_1,\dots,p_n)}$$ and show that it is at least $1/2$ for any $p_1,\dots,p_n$ (and similarly for revenue or any convex combination of welfare and revenue). Using the formula (\ref{eqn:cw}) for $c_w$ given in Lemma \ref{LEM:FORMULAS}, we can write the ratio in terms of the variables $A_i=\frac{\int_{x\geq p_i}\ell d\mu}{\int_{x\geq p_1}\ell d\mu}$ and $B_i=F(p_i)$ for $1\leq i\leq n$. For any fixed values of $B_i$, we then deduce the values of $A_i$ that minimize the ratio of interest. Finally, we show that the remaining expression is always at least $1/2$ no matter the values of $B_i$. The full proof can be found in Appendix \ref{app:proofapproxgeneral}.

\subsection{Tighter bounds for specific parameters}
\label{subsec:oneservertwolengths}

Assume in this subsection that there are jobs of two lengths $a<b$ which appear at each time step with probability $r_1$ and $r_2$, respectively, where $r_1+r_2\leq 1$. Suppose that we set a price per time step $p_1$ for jobs of length $a$ and $p_2$ for jobs of length $b$. Recall that the value per time step of a job is drawn from a distribution with cumulative distribution function $F$ and probability density function $f$. 

Our next result exhibits a tight approximation bound for any fixed setting of the job lengths and their distribution.

\begin{theorem}
\label{THM:TWOLENGTHSMAIN}
For any prices $p_1$ and $p_2$ that we set in the two-price setting, we can achieve a welfare (resp. revenue, or any convex combination of welfare and revenue) approximation of at least 
\begin{align*}
\rho(a,b,r_1,r_2)&:=\frac{(ar_1+br_2)(ar_1+1-r_1)}{a(a-1)r_1^2+a(b-1)r_1r_2+ar_1+br_2} \\
%&= \frac{a(a-1)r_1^2+b(a-1)r_1r_2+ar_1+br_2}{a(a-1)r_1^2+a(b-1)r_1r_2+ar_1+br_2}.
\end{align*}
in the one-price setting by setting either $p_1$ or $p_2$ alone. Moreover, this bound is the best possible even if we are allowed to set a price different from $p_1$ or $p_2$ in the one-price setting.
\end{theorem}

To prove this theorem, we work with the expression in terms of $B_i=F(p_i)$ that we have from the proof of Theorem \ref{THM:50PERCENTGENERAL}. We then show that the expression is minimized when we take $B_1=0$ and $B_2=1$, meaning that the distribution on job values is bimodal. The proof method readily yields an example showing that our bound is tight, where the bimodal distribution on job values puts a large probability on a low value and a small probability on a high value. The full proof can be found in Appendix \ref{app:twolengthsmain}.

Theorem \ref{THM:TWOLENGTHSMAIN} allows us to obtain the worst-case approximation ratio in arbitrary settings of the parameters. Some examples follow.

\begin{itemize}
\item Suppose that $r_1+r_2=1$, i.e., a job appears at every time step. This assumption can in fact be made without loss of generality, because we can convert jobs not arriving to jobs arriving with a value of 0 as long as all prices are nonzero. This only changes $F$ and so is irrelevant for the calculation of $\rho$. In this case, the approximation ratio is
$$\rho(a,b,r,1-r)=\frac{(ar+b-br)(ar+1-r)}{a(a-b)r^2+b(a-1)r+b}.$$
\item Next, we look at how the bound behaves when the job lengths are close together or far apart. Suppose for convenience that $r_1=r_2=1/2$, i.e., a job appears at every time step and is of length $a$ or $b$ with equal probability. Then the approximation ratio is
$$\rho\left(a,b,\frac{1}{2},\frac{1}{2}\right)=\frac{(a+b)(a+1)}{a^2+ab+2b}.$$
In particular, the approximation ratio is 
\begin{itemize}
\item $\frac{6}{7}$ if $a=1$ and $b=2$;
\item $\frac{4}{5}$ if $a=1$ and $b=3$;
\item $\frac{15}{16}$ if $a=2$ and $b=3$;
\item $\frac{a+1}{a+2}$ if $b\rightarrow\infty$.
\end{itemize}
Note that if $b\rightarrow\infty$, then the ratio approaches 1 as $a$ grows. This makes sense because when the two job lengths are large and close to each other, there is little difference between accepting one or the other. Also, if we take $b=a+1$, then the ratio becomes
$$\rho\left(a,a+1,\frac{1}{2},\frac{1}{2}\right)=\frac{2a^2+3a+1}{2a^2+3a+2}.$$ This also converges to 1 as $a\rightarrow\infty$.
\item We now consider the behavior of the bound when the shorter job length is fixed. Suppose for convenience that $a=1$, i.e., the shorter job length is 1. Then the approximation ratio is
$$\rho(1,b,r_1,r_2)=\frac{r_1+br_2}{(b-1)r_1r_2+r_1+br_2}.$$
As the longer job length grows, this ratio decreases and approaches $\frac{1}{1+r_1}$. This is consistent with the intuition that the approximation gets worse as the job lengths are farther apart.
\item We next look at the ``opposite'' of the previous case and assume that the longer job length is extremely large. Suppose that $b\rightarrow\infty$. Then the approximation ratio is
$$\rho(a,\infty,r_1,r_2)=\frac{ar_1+1-r_1}{ar_1+1}.$$
Note that this ratio does not depend on $r_2$. The ratio increases as the shorter job length grows. Again, this is consistent with the intuition that the approximation gets worse as the job lengths are farther apart.
\end{itemize}

If we fix the probabilities $r_1,r_2$, we can derive a tight worst-case bound over all possible job lengths $a,b$.

\begin{theorem}
\label{thm:twolengthsfixprob}
For fixed $r_1,r_2$, we have $$\rho(a,b,r_1,r_2)\geq\frac{1}{1+r_1}$$
for arbitrary $a,b$. Moreover, this bound is the best possible.
\end{theorem}

\begin{proof}
From Theorem \ref{THM:TWOLENGTHSMAIN}, we need to show the inequality
$$\frac{(ar_1+br_2)(ar_1+1-r_1)}{a(a-1)r_1^2+a(b-1)r_1r_2+ar_1+br_2}\geq\frac{1}{1+r_1}.$$
This simplifies to
$$r_1(a(a-1)r_1^2+(a-1)br_1r_2+ar_1+ar_2)\geq 0.$$
Since each term on the left-hand side is nonnegative, the inequality holds.

For tightness of the bound, let $a=1$ and $b\rightarrow\infty$. The approximation ratio is
\begin{align*}
\rho(1,\infty,r_1,r_2)&=\frac{r_1+br_2}{r_1(1-r_2)+br_2(1+r_1)} \\
&= \frac{1}{1+r_1},
\end{align*}
as desired.
\end{proof}

Note that the fact that the bound is tight at $a=1$ and $b\rightarrow\infty$ is consistent with the intuition that the further apart the job lengths are, the more welfare and revenue there is to be gained by setting different prices for different job lengths, and consequently the worse the approximation ratio.

Finally, we show that we can obtain at least $50\%$ of the welfare or revenue from setting two prices by using one of those prices.

\begin{theorem}
\label{thm:twolengthshalf}
For arbitrary $a,b,r_1,r_2$, we have
$$\rho(a,b,r_1,r_2)\geq\frac{1}{2}.$$
Moreover, this bound is the best possible.
\end{theorem}

\begin{proof}
Using Theorem \ref{thm:twolengthsfixprob}, we find that
$$\rho(a,b,r_1,r_2)\geq\frac{1}{1+r_1}\geq\frac{1}{2}$$
since $r_1\leq 1$.

For tightness of the bound, let $a=1$ and $r_2=1-r_1$. The approximation ratio is
$$\rho(1,b,r,1-r)=\frac{r+(1-r)b}{r^2+(1-r^2)b}.$$
Taking $b=\frac{1}{(1-r)^2}$, the ratio becomes
$$\frac{r+\frac{1}{1-r}}{r^2+\frac{1+r}{1-r}}=\frac{r-r^2+1}{r^2-r^3+1+r},$$
which approaches $1/2$ as $r\rightarrow 1^-$.\footnote{The approximation ratio does not necessarily converge to $1/2$ for an arbitrary direction as $b\rightarrow\infty$ and $r\rightarrow 1^-$. For instance, if we take $b=\frac{1}{1-r}$ and $r\rightarrow 1^-$, then the ratio converges to $2/3$.}
\end{proof}

While we do not have a general formula for the worst-case approximation ratio for each choice of the parameters $a_1,\dots,a_n,r_1,\dots,r_n$ as we do for the case of two job lengths, the function $h$ in the proof of Theorem \ref{THM:50PERCENTGENERAL} still allows us to derive a tighter bound for each specific case. Note that to find the minimum of $h$, it suffices to check $B_i=0$ or 1 (see Footnote \ref{footnote:zeroorone}), so we only have a finite number of cases to check. 

In fact, one can show that the minimum is always attained when $B_1=0$ and $B_n=1$. However, as we show next, the remaining $B_i$'s may be 0 or 1 at the minimum, depending on the job lengths and the probability distribution over them. We take $n=3$ and assume for convenience that $r_1=r_2=r_3=1/3$ and $(a_1,a_2)=(2,3)$. The following examples show that $B_2$ can be 0 or 1 at the minimum depending on how far the longest job length $a_3$ is from $a_1$ and $a_2$.

\begin{itemize}
\item Suppose that $a_3=6$. Then we have
$$h(B_1,B_2,B_3)=\frac{121-11B_1-22B_2-55B_3}{121-16B_1-24B_2-48B_3}.$$
This is minimized at $(B_1,B_2,B_3)=(0,1,1)$, where its value is $44/49$.
\item Suppose that $a_3=7$. Then we have
$$h(B_1,B_2,B_3)=\frac{48-4B_1-8B_2-24B_3}{48-6B_1-9B_2-21B_3}.$$ 
This is minimized at $(B_1,B_2,B_3)=(0,B_2,1)$ for any $0\leq B_2\leq 1$, where its value is $8/9$.
\item Suppose that $a_3=8$. Then we have
$$h(B_1,B_2,B_3)=\frac{169-13B_1-26B_2-91B_3}{169-20B_1-30B_2-80B_3}.$$
This is minimized at $(B_1,B_2,B_3)=(0,0,1)$, where its value is $78/89$.
\end{itemize}

The above examples show that the transition point where the optimal value of $B_2$ goes from 0 to 1 for $r_1=r_2=r_3=1/3$ and $(a_1,a_2)=(2,3)$ is at $a_3=7$, where $h(B_1,B_2,B_3)$ takes on the same value for any $0\leq B_2\leq 1$.

\subsection{Extension}
\label{subsec:extension}

In this subsection, we show that by using a single price, we can obtain $50\%$ of the welfare not only compared to using multiple prices, but also compared to the offline optimal welfare.\footnote{For the offline optimal welfare, we compute the limit of the expected average offline optimal welfare per time step as the time horizon grows.} In fact, we will also not need the assumption that the job length and the value per time step are independent. However, the result only works for particular prices rather than arbitrary ones, and we cannot obtain tighter results for specific parameters using this method.

\begin{theorem}
\label{THM:OFFLINEOPT}
Assume that the job length and the value per time step are not necessarily independent. There exists a price $p$ such that we can achieve a $50\%$ approximation of the offline optimal welfare by using $p$ as the single price.
\end{theorem}

The proof of Theorem~\ref{THM:OFFLINEOPT} can be found in Appendix~\ref{sec:offlineopt}.

\section{Multiple Servers}
\label{sec:multservers}

In this section, we assume that there are multiple servers, each of which receives jobs of various lengths. Under the assumption that the servers have the same probability of receiving no job at a time step, we show in Subsection \ref{subsec:multserversoneprice} an approximation bound of the welfare and revenue of setting one price for all servers compared to setting an individual price for each server. %\footnote{Using similar techniques, we also obtain approximation bounds when there is only one job length but a job may appear for each server with a different probability. The details are given in Appendix \ref{sec:multserversonelength}.} 
This yields a strong bound when at least one of the dimensions of the parameters is not too extreme, e.g., the number of servers or the job lengths are not too large. In Subsection \ref{subsec:multserversmultprice}, we combine the newly obtained results with those from Section \ref{sec:oneserver}. Using a composition technique, we derive a general result that compares the welfare and revenue obtained by a restricted mechanism that sets the same price for all servers and all job lengths against those obtained by a mechanism that can set a different price for each job length of each particular server. We show that even with the heavy restrictions, the former mechanism still provides a reasonable approximation to the latter in a wide range of situations.  Using similar techniques, we also obtain approximation bounds when this assumption does not hold but there is only one job length across all servers. The analysis of the latter setting is deferred to Appendix \ref{sec:multserversonelength}.

As in Section \ref{sec:oneserver}, our approximation results hold for arbitrary (i.e., not necessarily optimal) pricing schemes, and the price we use in the single-price setting can be drawn from one of the prices in the multi-price setting.

\subsection{One price per server}
\label{subsec:multserversoneprice}

Assume that at each time step, either zero or one job appears for each server $1\leq j\leq n$. Server $j$ receives jobs of length $a_{j1}\leq a_{j2}\leq\dots\leq a_{jn_j}$ with probability $r_{j1},r_{j2},\dots,r_{jn_j}$, respectively. Suppose that we set a price per time step $p_j$ for all jobs on server $j$. Recall that the value per time step of a job is drawn from a distribution with cumulative distribution function $F$ and probability density function $f$, and that we assume that $\sum_{i=1}^{n_j}r_{ji}$ is constant. Let $S_j=a_{j1}r_{j1}+\dots+a_{jn_j}r_{jn_j}$ and $R=r_{j1}+\dots+r_{jn_j}$.

Using the formula (\ref{eqn:cw}) for $c_w$ given in Lemma \ref{LEM:FORMULAS}, we find that the welfare per time step is
\begin{align*}
d_w(p_1,p_2,\dots,p_n)&=\sum_{j=1}^n\frac{S_j\int_{x\geq p_j}\ell d\mu}{S_j-(S_j-R)F(p_j)+(1-R)}\\
&=\sum_{j=1}^n\frac{\int_{x\geq p_j}\ell d\mu}{1-\left(1-\frac{R}{S_j}\right)F(p_j)+\frac{1-R}{S_j}}.
\end{align*}

If we set the same price $p=p_1=\dots=p_n$ for different servers, our welfare per time step becomes
$$d_w(p)=\sum_{j=1}^n\frac{\int_{x\geq p}\ell d\mu}{1-\left(1-\frac{R}{S_j}\right)F(p)+\frac{1-R}{S_j}}.$$

Similarly, we have the formulas for revenue per time step
$$d_r(p_1,p_2,\dots,p_n)=\sum_{j=1}^n\frac{(1-F(p_j))p_j}{1-\left(1-\frac{R}{S_j}\right)F(p_j)+\frac{1-R}{S_j}}$$
and
$$d_r(p)=\sum_{j=1}^n\frac{(1-F(p))p}{1-\left(1-\frac{R}{S_j}\right)F(p)+\frac{1-R}{S_j}}.$$

We show that if at least one dimension of the parameters is not too extreme, e.g., the number of servers or the job lengths are bounded, then we can obtain a reasonable approximation of the welfare and revenue in the multi-price setting by setting just one price.

\begin{theorem}
\label{THM:MULTSERVERSMULTLENGTHS}
For any prices $p_1,p_2,\dots,p_n$ that we set in the multi-price setting, we can achieve a welfare (resp. revenue, or any convex combination of welfare and revenue) approximation of at least 
$$\max\left(\frac{1}{H_n},\frac{M-1}{M\ln M}\right)$$
in the one-price setting, where $H_n=1+\frac{1}{2}+\dots+\frac{1}{n}\approx\ln n$ is the $n$th Harmonic number and $M=\max_{i,j}\frac{S_i}{S_j}$.
\end{theorem}

In particular, if all job lengths are bounded above by $c$, then $R\leq S_j\leq cR$ for all $1\leq j\leq n$, and so $\max_{i,j}\frac{S_i}{S_j}\leq c$. The theorem then implies that the approximation ratio is at least $\frac{c-1}{c\ln c}$.

The proof follows a similar outline to that of Theorem \ref{THM:50PERCENTGENERAL}, but the details are more involved. It can be found in Appendix \ref{app:multserversmultlengths}.

\subsection{Multiple prices per server}
\label{subsec:multserversmultprice}

Assume as in Subsection \ref{subsec:multserversoneprice} that at each time step, server $j$ receives jobs of length $a_{j1}\leq a_{j2}\leq\dots\leq a_{jn_j}$ with probability $r_{j1},r_{j2},\dots,r_{jn_j}$, respectively. In this subsection, we consider setting an individual price not only for each server but also for each job length of that server. In particular, suppose that we set a price per time step $p_{ji}$ for jobs of length $a_{ji}$ on server $j$. Recall that the value per time step of a job is drawn from a distribution with cumulative distribution function $F$ and probability density function $f$, and that we assume that $\sum_{i=1}^{n_j}r_{ji}$ is constant. Let $S_j=a_{j1}r_{j1}+\dots+a_{jn_j}r_{jn_j}$.

We will compare a setting where we have considerable freedom with our pricing scheme and can set a different price $p_{ji}$ for each job length $a_{ji}$ on each server $j$ with a setting where we have limited freedom and must set the same price $p$ for all job lengths and all servers. We show that by ``composing'' our results on the two dimensions, we can obtain an approximation of the welfare and revenue of setting different prices by setting a single price.

\begin{theorem}
\label{thm:multserversmultlengthsgeneral}
For any prices $p_{ji}$, where $1\leq j\leq n$ and $1\leq i\leq n_j$ for each $j$, that we set in the multi-price setting, we can achieve a welfare (resp. revenue, or any convex combination of welfare and revenue) approximation of at least 
$$\frac{1}{2}\cdot\max\left(\frac{1}{H_n},\frac{M-1}{M\ln M}\right)$$
in the one-price setting, where $H_n=1+\frac{1}{2}+\dots+\frac{1}{n}\approx\ln n$ is the $n$th Harmonic number and $M=\max_{i,j}\frac{S_i}{S_j}$.
\end{theorem}

\begin{proof}
Consider welfare. By Theorem \ref{THM:50PERCENTGENERAL}, for each server $j$ we can achieve an $\frac{1}{2}$-approximation of the welfare in the multi-price setting by setting a single price $p_j$ for all job lengths. On the other hand, using Theorem \ref{THM:MULTSERVERSMULTLENGTHS}, we can approximate the latter welfare by a factor of $\max\left(\frac{1}{H_n},\frac{M-1}{M\ln M}\right)$ by setting a single price $p$ for all servers. Therefore, setting a single price $p$ also yields a $\frac{1}{2}\cdot\max\left(\frac{1}{H_n},\frac{M-1}{M\ln M}\right)$-approximation of the original welfare.

The same argument holds for revenue and for any convex combination of welfare and revenue.
\end{proof}

If we have tighter approximations for either the ``different prices for different job lengths'' or the ``different prices for different servers'' dimension, for instance by knowing the values of some of the parameters, then the same composition argument yields a correspondingly tighter bound.

\section{Conclusion}

In this paper, we study how well simple pricing schemes that are oblivious to certain parameters can approximate optimal schemes with respect to welfare and revenue, and prove several results when the simple schemes are restricted to setting the same price for all servers or all job lengths. Our results provide an explanation of the efficacy of such schemes in practice, including the one shown in Figure \ref{fig:azure} for virtual machines on Microsoft Azure. Since simple schemes do not require agents to spend time and resources to determine their specific parameter values, our results also serve as an argument in favor of using these schemes in a range of applications. It is worth noting that as all of our results are of worst case nature, we can expect the guarantees on welfare and revenue to be significantly better than these pessimistic bounds in practical instances where the parameters are not adversarially tailored.

We believe that there is still much interesting work to be done in the study of simple pricing schemes for the cloud. We conclude our paper by listing some intriguing future directions.

\begin{itemize}
\item In many scheduling applications, a job can be scheduled online to any server that is not occupied at the time. Does a good welfare or revenue approximation hold in such a model?

\item Can our results be extended to models with more fluid job arrivals, for example one where several jobs can arrive at each time step?

\item Can we approximate welfare and revenue simultaneously? A trivial randomized approach would be to choose with equal probability whether to approximate welfare or revenue. According to Theorem \ref{THM:50PERCENTGENERAL}, this yields a $1/4$-approximation for both expected welfare and expected revenue of the single-price setting in comparison to the multi-price setting for job lengths.
\end{itemize}

\section*{Acknowledgments}

Preliminary versions of this paper appeared in Proceedings of the 13th Conference on Web and Internet Economics, December 2017, and Proceedings of the 12th Workshop on the Economics of Networks, Systems and Computation, June 2017. We thank the anonymous reviewers for helpful comments. Warut Suksompong is partially supported by a Stanford Graduate Fellowship.

% Bibliography
\bibliographystyle{ACM-Reference-Format}
\bibliography{main}

% Appendix
\appendix

\section{Proof of Lemma \ref{LEM:FORMULAS}}
\label{app:prooflemmaformulas}

We represent the states of the server by a Markov chain. Initially, we have an idle state corresponding to when the server is free. For each job length $a_i$, we create $a_i-1$ states that the server goes through when it accepts a job of length $a_i$. The states represent the number of time steps remaining to complete the service of the job before the server returns to the idle state. The rewards, which can be either welfare or revenue, are collected at each transition. The structure of the Markov chain makes solving for the stationary distribution straightforward.  This distribution gives the average proportion of time spent in each transition and the average welfare or revenue can be written in terms of these proportions. If the threshold for accepting a job of a certain length is $p$, then the revenue collected for each transition with that job length is always $p$, while the expected welfare gained during the transition is given by $\frac{\int_{x\geq p}\ell d\mu}{1-F(p)}$. This allows us to write down the formulas for $c_w$ and $c_r$.

Next, we present a simpler approach that is perhaps less formal. Consider an arbitrary time step when the server is free. If the job drawn at that time step has length $a_1$, then with probability $F(p_1)$ it has value below $p_1$ and is rejected, and with probability $1-F(p_1)$ it has value at least $p_1$ and is accepted. In the latter case, the job yields expected welfare $\frac{\int_{x\ge p_1}\ell d\mu}{1-F(p_1)}$ per time step over $a_1$ time steps. Analogous statements hold if the job has length $a_i$ for $2\leq i\leq n$. It follows that
\begin{equation*}
\begin{split}
0&=r_1\left(-F(p_1)c_w+(1-F(p_1))\left(a_1\cdot\frac{\int_{x\geq p_1}\ell d\mu}{1-F(p_1)}-a_1c_w\right)\right) \\ 
&\quad +r_2\left(-F(p_2)c_w+(1-F(p_2))\left(a_2\cdot\frac{\int_{x\geq p_2}\ell d\mu}{1-F(p_2)}-a_2c_w\right)\right) \\
&\quad +\dots \\
&\quad +r_n\left(-F(p_n)c_w+(1-F(p_n))\left(a_n\cdot\frac{\int_{x\geq p_n}\ell d\mu}{1-F(p_n)}-a_nc_w\right)\right) \\
&\quad +(1-r_1-r_2-\dots-r_n)(-c_w).
\end{split}
\end{equation*}
Solving for $c_w$, we have
$$c_w(p_1,p_2,\dots,p_n)=\frac{a_1r_1\int_{x\geq p_1}\ell d\mu+\dots+a_nr_n\int_{x\geq p_n}\ell d\mu}{S-((a_1-1)r_1F(p_1)+\dots+(a_n-1)r_nF(p_n))+(1-R)}.$$

For revenue, we can derive the equations in the same way, with the exception that the revenue from accepting a job at price $p$ is simply $p$. We have
\begin{equation*}
\begin{split}
0&=r_1\left(-F(p_1)c_r+(1-F(p_1))\left(a_1p_1-a_1c_r\right)\right) \\
&\quad +r_2\left(-F(p_2)c_r+(1-F(p_2))\left(a_2p_2-a_2c_r\right)\right)\\
&\quad +\dots\\
&\quad +r_n\left(-F(p_n)c_r+(1-F(p_n))\left(a_np_n-a_nc_r\right)\right)\\
&\quad +(1-r_1-r_2-\dots-r_n)(-c_r).
\end{split}
\end{equation*}
Solving for $c_r$, we get
$$c_r(p_1,p_2\dots,p_n)=\frac{a_1r_1(1-F(p_1))p_1+\dots+a_nr_n(1-F(p_n))p_n}{S-((a_1-1)r_1F(p_1)+\dots+(a_n-1)r_nF(p_n))+(1-R)}.$$

\section{Proof of Theorem \ref{THM:50PERCENTGENERAL}}
\label{app:proofapproxgeneral}

We first consider welfare. We wish to show that setting one of the prices $p_i$ alone achieves an approximation of $1/2$ of setting all $n$ prices. That is,
$$\max(c_w(p_1),\dots,c_w(p_n))\geq\frac{1}{2} \cdot c_w(p_1,\dots,p_n).$$
To establish this inequality, we will work with the ratio
$$\frac{\max(c_w(p_1),\dots,c_w(p_n))}{c_w(p_1,\dots,p_n)}$$
and show that its minimum is at least $1/2$.

Writing $A_i=\frac{\int_{x\geq p_i}\ell d\mu}{\int_{x\geq p_1}\ell d\mu}$ (in particular, $A_1=1$) and $B_i=F(p_i)$ for $1\leq i\leq n$, the ratio to minimize becomes
\begin{multline*}
g(A_1,\dots,A_n,B_1,\dots,B_n):=\max_{i=1}^n\left(\frac{SA_i}{S-(S-R)B_i+(1-R)}\right) \\
\cdot\frac{S-((a_1-1)r_1B_1+\dots+(a_n-1)r_nB_n)+(1-R)}{a_1r_1A_1+\dots+a_nr_nA_n}.
\end{multline*}

\emph{Case 1}: The max function outputs the first term, $\frac{SA_1}{S-(S-R)B_1+(1-R)}$.

Taking into account that $A_1=1$, we want to minimize the ratio
$$\frac{S}{S-(S-R)B_1+(1-R)}\cdot\frac{S-((a_1-1)r_1B_1+\dots+(a_n-1)r_nB_n)+(1-R)}{a_1r_1+a_2r_2A_2+\dots+a_nr_nA_n},$$
where $A_i\leq\frac{S-(S-R)B_i+(1-R)}{S-(S-R)B_1+(1-R)}$ for $i\geq 2$.

For any $A_i$, if we fix the remaining $A_j$ and all $B_j$, this is a decreasing function in $A_i$. To minimize it, we should set $A_i=\frac{S-(S-R)B_i+(1-R)}{S-(S-R)B_1+(1-R)}$ for all $i\geq 2$. The ratio becomes
$$\frac{S^2-((a_1-1)r_1B_1+\dots+(a_n-1)r_nB_n)S+S(1-R)}{S^2-(a_1r_1B_1+\dots+a_nr_nB_n)(S-R)+S(1-R)}.$$

\emph{Case 2}: The max function outputs the $i$th term, $\frac{SA_i}{S-(S-R)B_i+(1-R)}$, for some $i\geq 2$. This means that $A_j\leq\frac{S-(S-R)B_j+(1-R)}{S-(S-R)B_i+(1-R)}\cdot A_i$ for all $j\geq 2$ with $j\neq i$, and $A_i\geq\frac{S-(S-R)B_i+(1-R)}{S-(S-R)B_1+(1-R)}$. We want to minimize the ratio
$$\frac{SA_i}{S-(S-R)B_i+(1-R)}\cdot\frac{S-((a_1-1)r_1B_1+\dots+(a_n-1)r_nB_n)+(1-R)}{a_1r_1+a_2r_2A_2+\dots+a_nr_nA_n},$$
or equivalently,
$$\frac{SA_i}{a_1r_1A_1+\dots+a_nr_nA_n}\cdot\frac{S-((a_1-1)r_1B_1+\dots+(a_n-1)r_nB_n)+(1-R)}{S-(S-R)B_i+(1-R)}.$$

For any $j\geq 2$ with $j\neq i$, if we fix all terms $A_k$ except $A_j$ and fix all $B_k$, then this function is decreasing in $A_j$. To minimize it, we should set $A_j=\frac{S-(S-R)B_j+(1-R)}{S-(S-R)B_i+(1-R)}\cdot A_i$. The resulting function is increasing in $A_i$ if we fix all $B_k$, so we should set $A_i=\frac{S-(S-R)B_i+(1-R)}{S-(S-R)B_1+(1-R)}$. We obtain the same ratio as in Case 1.

Hence in either case, we are left with minimizing the function
$$h(B_1,\dots,B_n):=\frac{S^2-((a_1-1)r_1B_1+\dots+(a_n-1)r_nB_n)S+S(1-R)}{S^2-(a_1r_1B_1+\dots+a_nr_nB_n)(S-R)+S(1-R)}.$$
In particular, we want to show that $h(B_1,\dots,B_n)\geq\frac{1}{2}$ for any choice of $B_1,\dots,B_n$. This is equivalent to
$$S^2+S(1-R)\geq(2S(a_1-1)-a_1(S-R))r_1B_1+\dots+(2S(a_n-1)-a_n(S-R))r_nB_n$$
or
$$S^2+S(1-R)\geq((a_1-2)S+a_1R)r_1B_1+\dots+((a_n-2)S+a_nR)r_nB_n.$$

We consider two cases.

\begin{enumerate}
\item $a_1\geq 2$ (and hence $a_2,\dots,a_n\geq 2$). All coefficients of $r_iB_i$ on the right-hand side are positive, so we only need to verify the inequality for $B_1=\dots=B_n=1$. We have 
\begin{align*}
h(B_1=1,\dots,B_n=1)&=\frac{S^2-((a_1-1)r_1+\dots+(a_n-1)r_n)S+S(1-R)}{S^2-(a_1r_1+\dots+a_nr_n)(S-R)+S(1-R)} \\
&= \frac{S^2-(S-R)S+S(1-R)}{S^2-S(S-R)+S(1-R)}=1>\frac{1}{2}.
\end{align*}
\item $a_1=1$.\footnote{The proof proceeds similarly if there are other $a_i$'s equal to 1.} All coefficients of $r_iB_i$ on the right-hand side except the first one are positive, so we only need to verify the inequality for $B_1=0$ and $B_2=\dots=B_n=1$. We have
\begin{align*}
h(B_1=0,B_2=1,\dots,B_n=1)&=\frac{S}{S+r_1(S-R)}> \frac{1}{2},
\end{align*}
where the last line follows from $r_1(S-R)<1\cdot S=S$.
\end{enumerate}

So the inequality holds in both cases, and the approximation ratio is at least $\frac{1}{2}$, as claimed.

Finally, we can obtain analogous results for revenue by essentially repeating the same argument but instead writing $A_i=\frac{(1-F(p_i))p_i}{(1-F(p_1))p_1}$ for $1\leq i\leq n$, and for any convex combination of welfare and revenue by writing $A_i$ as the appropriate convex combination of the two corresponding terms.

\section{Proof of Theorem \ref{THM:TWOLENGTHSMAIN}}
\label{app:twolengthsmain}

Using the proof of Theorem \ref{THM:50PERCENTGENERAL}, we are left with minimizing the function
$$h(B_1,B_2):=\frac{S^2-((a-1)r_1B_1+(b-1)r_2B_2)S+S(1-R)}{S^2-ar_1(S-R)B_1-br_2(S-R)B_2+S(1-R)}.$$

Note that the numerator and the denominator are positive for all $0\leq B_1,B_2\leq 1$. When $B_1=B_2=B$, both terms are equal to 
$$S^2-S(S-R)B+S(1-R),$$ 
and so $h(B_1,B_2)=1$. Moreover, using the assumption $a<b$, we find that 
\begin{itemize}
\item $h(B_1,B_2)<1$ when $B_1=0$ and $B_2>0$, and
\item $h(B_1,B_2)>1$ when $B_1>0$ and $B_2=0$.
\end{itemize}
Hence the function $h(B_1,B_2)$ is increasing in $B_1$ for fixed $B_2$, and decreasing in $B_2$ for fixed $B_1$.\footnote{\label{footnote:zeroorone}To see this, observe that an arbitrary function of the form $f(x)=\frac{p+qx}{r+sx}$ for real constants $p,q,r,s$ is either increasing, decreasing, or constant in any interval where the denominator is nonzero, depending on whether $qr-ps$ is positive, negative, or zero, respectively.} This implies that the function is minimized when $B_1=0$ and $B_2=1$, where its value is
\begin{align*}
h(B_1=0,B_2=1)&=\frac{S^2-(b-1)r_2S+S(1-R)}{S^2-br_2(S-R)+S(1-R)} \\
&=\frac{(ar_1+br_2)(ar_1+1-r_1)}{a(a-1)r_1^2+a(b-1)r_1r_2+ar_1+br_2} \\
&= \rho(a,b,r_1,r_2),
\end{align*}
as claimed.

Next, we show that the approximation ratio is tight even if we are allowed to set an arbitrary price (i.e., not necessarily $p_1$ or $p_2$) in the one-price setting. To this end, consider a discrete bimodal distribution where a high probability $q_1\approx 1$ is put on a value $v_1\approx 0$ and a small probability $q_2\approx 0$ is put on a value $v_2\approx 1$.\footnote{Intuitively, we want to accept all short jobs but only high-valued long jobs. We can do this in the two-price setting, while in the one-price setting we are forced to either accept low-valued long jobs or reject low-valued short jobs.} The values $v_1,v_2$ and the probabilities are chosen arbitrarily close to 0 and 1 and so that the relation 
$$\frac{q_2v_2}{q_1v_1}=\frac{1}{S-R}$$
is satisfied.

In the two-price setting, we can set prices $p_1=v_1$ and $p_2=v_2$ and obtain welfare
\begin{align*}
c_w(p_1=v_1,p_2=v_2)&=\frac{ar_1(q_1v_1+q_2v_2)+br_2(q_2v_2)}{S-((a-1)r_1F(v_1)+(b-1)r_2F(v_2))+(1-R)} \\
&\approx\frac{ar_1(q_1v_1+q_2v_2)+br_2(q_2v_2)}{S-(b-1)r_2+(1-R)} \\
&=\frac{q_2v_2(ar_1(S+1-R)+br_2)}{S-(b-1)r_2+(1-R)}.
\end{align*}

On the other hand, in the one-price setting there are three price ranges that we can pick: $[0,v_1]$, $(v_1,v_2]$, and $(v_2,\infty)$. If we set a price in the range $(v_2,\infty)$, no job is accepted and the welfare is zero. For each of the other two ranges, setting any price in the range yields the same set of accepted jobs and thus the same welfare. Hence it suffices to consider setting prices $v_1$ and $v_2$. We have
$$c_w(p=v_1)=\frac{S(q_1v_1+q_2v_2)}{S+1-R}=Sq_2v_2$$
and
$$c_w(p=v_2)\approx\frac{Sq_2v_2}{S-(S-R)+(1-R)}=Sq_2v_2.$$

We obtain the same welfare per time step in either case. It follows that the approximation ratio is at most
\begin{align*}
\frac{c_w(p=v_1)}{c_w(p_1=v_1,p_2=v_2)} &= \frac{Sq_2v_2(S-(b-1)r_2+(1-R))}{q_2v_2(ar_1(S+1-R)+br_2)} \\
&= \frac{S(ar_1+1-r_1)}{a(a-1)r_1^2+a(b-1)r_1r_2+ar_1+br_2} \\
&= \rho(a,b,r_1,r_2).
\end{align*}

Finally, we can obtain analogous results for revenue by essentially repeating the same argument but instead writing $A_2=\frac{(1-F(p_2))p_2}{(1-F(p_1))p_1}$ in the proof of Theorem \ref{THM:50PERCENTGENERAL}, and for any convex combination of welfare and revenue by writing $A_2$ as the appropriate convex combination of the two corresponding terms.

\section{Proof of Theorem \ref{THM:OFFLINEOPT}}
\label{sec:offlineopt}

We prove the result for a distribution with discrete support and extend it to continuous and mixed support later.  Our approach follows that of \citet{FeldmanGrLu15}.
As job length and value are not independent, assume that there are jobs of classes $1,2,\dots,n$ and that a job of class $j$ arrives with probability $r_j$, has length $a_j$ and value per timestep $v_j$. Note that $a_j$ and $v_j$ need not be different for different classes of jobs.

We can write an ``expected LP'' to upper bound the maximum welfare per timestep as follows:
\begin{align*}
Opt &= \max \sum_j x_jv_ja_j \mbox{ s.t.}\\
x_j &\leq r_j\\
\sum_j x_ja_j &\leq 1
\end{align*}

Here, $Opt$ is the optimal long-run average welfare per time step, and the ratio $x_j/r_j$ can be thought of as the probability that an arriving job of class $j$ is accepted. The constraints then say that this probability is at most 1 and that we cannot accept more jobs in expectation than, if they arrived perfectly, would saturate the machine.

Take $p = Opt/2$,and consider some time $t$.  Let $y_t$ be the probability that the server is occupied at time $t$ (possibly by the job arriving at time $t$).  Then the expected revenue at time $t$ is $py_t = (Opt/2)\cdot y_t$.  The expected consumer surplus of the job arriving at time $t$ is at least $\sum_j r_j(v_j-p)a_j(1-y_t)$, because the probability that the server is unoccupied at time $t$ (either by the job arriving at time $t$ or by an earlier job) is a lower bound on the probability that the server was not occupied when the job at time $t$ arrived.  By the constraints of the LP, the consumer surplus is at least $(Opt/2)\cdot(1-y_t)$. Welfare is the sum of revenue and consumer surplus, so summing over all $t$ shows that the welfare per time step is at least $Opt/2$.

If the support of the distribution is not discrete, a similar argument still works. To determine $Opt$, we order the jobs in decreasing order of value per timestep and take the jobs until the machine is saturated. The rest of the proof then follows in the same way as before, with sums replaced by integrals.

As a further extension, this argument can also be adapted to the multiple server case.  We can solve the expected LP for each server individually and compute a price $p_i = Opt_i/2$.  We can then select whichever of these prices maximizes welfare to recover the $1/H_n$ bound from Theorem~\ref{thm:multserversmultlengthsgeneral}. (For some $i$, the welfare from the $i$ servers with the largest values of $Opt_i$ at price $p_i$, which we can lower bound as $i Opt_i/2$, must be at least $\frac{1}{2H_n} \sum_i Opt_i$.)

\section{Multiple Servers, One Job Length}
\label{sec:multserversonelength}

Assume that at each time step, either zero or one job appears for each server. Server $j$ receives a job of length $a$ with probability $r_j$. Suppose that we set a price per time step $p_j$ for jobs on server $j$. Recall that the value per time step of a job is drawn from a distribution with cumulative distribution function $F$ and probability density function $f$.

Using the formula (\ref{eqn:cw}) for $c_w$ given in Lemma \ref{LEM:FORMULAS}, we find that the welfare per time step is
\begin{align*}
e_w(p_1,p_2,\dots,p_n)&=\sum_{j=1}^n\frac{ar_j\int_{x\geq p_j}\ell d\mu}{ar_j-(a-1)r_jF(p_j)+(1-r_j)}\\
&=\sum_{j=1}^n\frac{a\int_{x\geq p_j}\ell d\mu}{(a-1)(1-F(p_j))+\frac{1}{r_j}}.
\end{align*}

If we set the same price $p=p_1=\dots=p_n$ for different servers, our welfare per time step becomes
$$e_w(p)=\sum_{j=1}^n\frac{a\int_{x\geq p}\ell d\mu}{\frac{1}{r_j}+(a-1)(1-F(p))}.$$

Similarly, we have the formulas for revenue per time step
\begin{align*}
e_r(p_1,p_2,\dots,p_n)
&=\sum_{j=1}^n\frac{a(1-F(p_j))p_j}{(a-1)(1-F(p_j))+\frac{1}{r_j}}
\end{align*}
and
$$e_r(p)=\sum_{j=1}^n\frac{a(1-F(p))p}{(a-1)(1-F(p))+\frac{1}{r_j}}.$$

We will compare the welfare and revenue that can be obtained by setting a single price against setting several prices. Similarly to Section \ref{sec:multservers}, we will show that if at least one dimension of the parameters is not too extreme, e.g., the number of servers, the job length, or the probabilities of jobs occurrence are bounded, then we can obtain a reasonable approximation of the welfare and revenue in the multi-price setting by setting just one price.

\begin{theorem}
\label{thm:multserversonelength}
For any prices $p_1,p_2,\dots,p_n$ that we set in the multi-price setting, we can achieve a welfare (resp. revenue, or any convex combination of welfare and revenue) approximation of at least 
$$\max\left(\frac{1}{H_n},\frac{M-1}{M\ln M},\frac{1}{a}\right)$$
in the one-price setting, where $H_n=1+\frac{1}{2}+\dots+\frac{1}{n}\approx\ln n$ is the $n$th Harmonic number and $M=\max_{i,j}\frac{r_i}{r_j}$.
\end{theorem}

In particular, if all probabilities of job occurrence are bounded below by $c\leq 1$, then $\max_{i,j}\frac{S_i}{S_j}\leq \frac{1}{c}$. The theorem then implies that the approximation ratio is also at least $\frac{\frac{1}{c}-1}{\frac{1}{c}\ln\left(\frac{1}{c}\right)}=\frac{c-1}{\ln c}$.

\begin{proof}
We first consider welfare. We will work with the ratio
$$\frac{\max(e_w(p_1),\dots,e_w(p_n))}{e_w(p_1,\dots,p_n)}$$ 
and try to minimize it.

Writing $A_j=\frac{\int_{x\geq p_j}\ell d\mu}{\int_{x\geq p_1}\ell d\mu}$ (in particular, $A_1=1$), $B_j=F(p_j)$, and $R_j=1/r_j$ for $1\leq j\leq n$, the ratio to minimize becomes
\begin{equation*}
g(A_1,\dots,A_n,B_1,\dots,B_n):=
\max_{j=1}^n\left(\sum_{i=1}^n\frac{A_j}{R_i+(a-1)(1-B_j)}\right) 
\cdot\frac{1}{\sum_{j=1}^n\frac{A_i}{R_i+(a-1)(1-B_i)}}.
\end{equation*}

Let $U_j=\sum_{i=1}^n\frac{1}{R_i+(a-1)(1-B_j)}$ for $1\leq j\leq n$, and $T=\sum_{i=1}^n\frac{A_i}{R_i+(a-1)(1-B_i)}$. We have
$$g(A_1,\dots,A_n,B_1,\dots,B_n)=\frac{\max_{j=1}^n(A_jU_j)}{T}.$$

\emph{Case 1}: The max function outputs the first term, $A_1U_1$. 

Taking into account that $A_1=1$, we want to minimize the ratio
$$\frac{U_1}{T}=\left(\sum_{i=1}^n\frac{1}{R_i+(a-1)(1-B_1)}\right) 
\cdot\frac{1}{\sum_{i=1}^n\frac{A_i}{R_i+(a-1)(1-B_i)}},$$
where $A_i\leq\frac{U_1}{U_i}$ for all $i\geq 2$.

For any $A_j$, if we fix the remaining $A_i$ and all $B_i$, this is a decreasing function in $A_j$. To minimize it, we should set $A_i=\frac{U_1}{U_i}$ for all $i\geq 2$. The ratio becomes
$$\frac{1}{\sum_{j=1}^n\frac{1}{\sum_{i=1}^n\frac{R_j+(a-1)(1-B_j)}{R_i+(a-1)(1-B_j)}}}.$$

\emph{Case 2}: The max function outputs the $j$th term, $A_jU_j$, for some $j\geq 2$.

This means that $A_i\leq\frac{A_jU_j}{U_i}$ for all $i\geq 2$ with $i\neq j$, and $A_j\geq\frac{U_1}{U_j}$. We want to minimize the ratio
$$\frac{A_jU_j}{T}=\left(\sum_{i=1}^n\frac{A_j}{R_i+(a-1)(1-B_j)}\right) 
\cdot\frac{1}{\sum_{i=1}^n\frac{A_i}{R_i+(a-1)(1-B_i)}}.$$

For any $i\geq 2$ with $i\neq j$, if we fix all terms $A_k$ except $A_i$ and fix all $B_k$, then this function is decreasing in $A_i$. To minimize it, we should set $A_i=\frac{A_jU_j}{U_i}$. The resulting function is increasing in $A_j$ if we fix all $B_k$, so we should set $A_j=\frac{U_1}{U_j}$. We obtain the same ratio as in Case 1.

Hence in either case, we are left with minimizing the function
$$h(B_1,\dots,B_n):=\frac{1}{\sum_{j=1}^n\frac{1}{\sum_{i=1}^n\frac{R_j+(a-1)(1-B_j)}{R_i+(a-1)(1-B_j)}}}.$$
In other words, the reciprocal of this function, which we want to maximize, is
\begin{align*}
h_0(B_1,\dots,B_n)&:=\sum_{j=1}^n\frac{1}{\sum_{i=1}^n\frac{R_j+(a-1)(1-B_j)}{R_i+(a-1)(1-B_j)}}.
\end{align*}

Assume without loss of generality that $R_1\geq R_2\geq\ldots\geq R_n$. For $1\leq j\leq n$, write $h_j=\frac{1}{\sum_{i=1}^n\frac{R_j+(a-1)(1-B_j)}{R_i+(a-1)(1-B_j)}}$, i.e., the $j$th term of the sum that constitutes $h_0$. The last $n+1-j$ terms of the sum in the denominator of $h_j$ are at least 1 since $R_j\geq R_i$ for $i\geq j$. This implies that $h_j\leq\frac{1}{n+1-j}$, and therefore 
\begin{align*}
h_0(B_1,\dots,B_n)&=h_1+h_2+\dots+h_n\\
&\leq \frac{1}{n}+\frac{1}{n-1}+\dots+1=H_n.
\end{align*}
Equivalently, $h(B_1,\dots,B_n)\geq\frac{1}{H_n}$.

Next, since $M=\max_{i,j}\frac{r_i}{r_j}$, we have $\frac{R_i}{R_j}\leq M$ for all $i,j$. Note that
$$\frac{R_i+(a-1)(1-B_i)}{R_j+(a-1)(1-B_i)}\geq \frac{1}{M}$$
for all $i,j$. This implies that
$$h_j\leq\frac{1}{n+1-j+\frac{j-1}{M}}$$
for all $1\leq j\leq n$, and therefore
\begin{align*}
h_0(B_1,\dots,B_n)&=\sum_{j=1}^n\frac{1}{n-(j-1)\left(1-\frac{1}{M}\right)}\\
&=\sum_{j=1}^n\frac{1}{n}\cdot\frac{1}{1-\frac{j-1}{n}\left(1-\frac{1}{M}\right)}.
\end{align*}
The last sum is the left Riemann sum of the function $\frac{1}{1-x\left(1-\frac{1}{M}\right)}$. This function is increasing in $x$, so its integral between 0 and 1 is an upper bound for our sum. Hence we have
$$h_0(B_1,\dots,B_n)\leq\int_0^1\frac{1}{1-x\left(1-\frac{1}{M}\right)}dx=\frac{M\ln M}{M-1}.$$
Equivalently, $h(B_1,\dots,B_n)\geq\frac{M-1}{M\ln M}$.

Finally, one can check that since $R_i\geq 1$ and $1-B_i\leq 1$ for all $i$,
$$\frac{R_i+(a-1)(1-B_i)}{R_j+(a-1)(1-B_i)}\geq \frac{1}{a}\cdot\frac{R_i}{R_j}$$
for all $i,j$. It follows that
\begin{align*}
h_0(B_1,\dots,B_n)&\leq\sum_{i=1}^n\frac{1}{\sum_{j=1}^n\frac{1}{a}\cdot\frac{R_i}{R_j}}\\
&=a\sum_{i=1}^n\frac{1}{\sum_{j=1}^n\frac{R_i}{R_j}}\\
&=a.
\end{align*}
Equivalently, $h(B_1,\dots,B_n)\geq\frac{1}{a}$.

Combining the three bounds, we have 
$$h(B_1,\dots,B_n)\geq\max\left(\frac{1}{H_n},\frac{M-1}{M\ln M},\frac{1}{a}\right),$$
as desired.

Finally, we can obtain analogous results for revenue by essentially repeating the same argument but instead writing $A_j=\frac{(1-F(p_j))p_j}{(1-F(p_1))p_1}$ for $1\leq j\leq n$, and for any convex combination of welfare and revenue by writing $A_j$ as the appropriate convex combination of the two corresponding terms.
\end{proof}

We now address the tightness of the approximation ratio. The upper bound $H_n$ for $h_0$ is the best possible in the sense that there exist values $B_1,\dots,B_n,R_1,\dots,R_n,a$ such that $h_0$ gets arbitrarily close to $H_n$. In particular, take $R_j=c^{2(n-j)}$ and $B_j=1-\frac{c^{2(n-j)+1}}{a-1}$ for some large constant $c$. We have
\begin{align*}
h_0(B_1,\dots,B_n)&=\sum_{j=1}^n\frac{1}{\sum_{i=1}^n\frac{R_j+(a-1)(1-B_j)}{R_i+(a-1)(1-B_j)}}\\
&=\sum_{j=1}^n\frac{1}{\sum_{i=1}^n\frac{c^{2(n-j)}+c^{2(n-j)+1}}{c^{2(n-i)}+c^{2(n-j)+1}}}\\
&=\sum_{j=1}^n\frac{1}{\sum_{i=1}^n\frac{1+c}{c^{2(j-i)}+c}}.
\end{align*}
Taking $c\rightarrow\infty$, we find that or $j\leq i$, the fraction $\frac{1+c}{c^{2(j-i)}+c}$ converges to 1, while for $j>i$, it converges to 0. Hence $h_j\rightarrow\frac{1}{n+1-j}$, and consequently $h_0\rightarrow H_n$. 

While this argument does not directly imply the tightness of the approximation ratio, we see it as strong evidence for that claim.

\section{Proof of Theorem \ref{THM:MULTSERVERSMULTLENGTHS}}
\label{app:multserversmultlengths}

We first consider welfare. We will work with the ratio
$$\frac{\max(d_w(p_1),\dots,d_w(p_n))}{d_w(p_1,\dots,p_n)}$$ 
and try to minimize it.

Writing $A_j=\frac{\int_{x\geq p_j}\ell d\mu}{\int_{x\geq p_1}\ell d\mu}$ (in particular, $A_1=1$), $B_j=F(p_j)$, and $s_j=\frac{1}{S_j}$ for $1\leq j\leq n$, the ratio to minimize becomes
\begin{equation*}
g(A_1,\dots,A_n,B_1,\dots,B_n):=\\
\max_{j=1}^n\left(\sum_{i=1}^n\frac{A_j}{1-(1-Rs_i)B_j+(1-R)s_i}\right) 
\cdot\frac{1}{\sum_{i=1}^n\frac{A_i}{1-(1-Rs_i)B_i+(1-R)s_i}}.
\end{equation*}

Let $U_j=\sum_{i=1}^n\frac{1}{1-(1-Rs_i)B_j+(1-R)s_j}$ for $1\leq j\leq n$, and $T=\sum_{i=1}^n\frac{A_i}{1-(1-Rs_i)B_i+(1-R)s_i}$. We have
$$g(A_1,\dots,A_n,B_1,\dots,B_n)=\frac{\max_{j=1}^n(A_jU_j)}{T}.$$

\emph{Case 1}: The max function outputs the first term, $A_1U_1$. 

Taking into account that $A_1=1$, we want to minimize the ratio
$$\frac{U_1}{T}=\left(\sum_{i=1}^n\frac{1}{1-(1-Rs_i)B_1+(1-R)s_i}\right) 
\cdot\frac{1}{\sum_{i=1}^n\frac{A_i}{1-(1-Rs_i)B_i+(1-R)s_i}},$$
where $A_i\leq\frac{U_1}{U_i}$ for all $i\geq 2$.

For any $A_j$, if we fix the remaining $A_i$ and all $B_i$, this is a decreasing function in $A_j$. To minimize it, we should set $A_i=\frac{U_1}{U_i}$ for all $i\geq 2$. The ratio becomes
$$\frac{1}{\sum_{j=1}^n\frac{1}{\sum_{i=1}^n\frac{1-(1-Rs_j)B_j+(1-R)s_j}{1-(1-Rs_i)B_j+(1-R)s_i}}}.$$

\emph{Case 2}: The max function outputs the $j$th term, $A_jU_j$, for some $j\geq 2$.

This means that $A_i\leq\frac{A_jU_j}{U_i}$ for all $i\geq 2$ with $i\neq j$, and $A_j\geq\frac{U_1}{U_j}$. We want to minimize the ratio
$$\frac{A_jU_j}{T}=\left(\sum_{i=1}^n\frac{A_j}{1-(1-Rs_i)B_j+(1-R)s_i}\right) 
\cdot\frac{1}{\sum_{i=1}^n\frac{A_i}{1-(1-Rs_i)B_i+(1-R)s_i}}.$$

For any $i\geq 2$ with $i\neq j$, if we fix all terms $A_k$ except $A_i$ and fix all $B_k$, then this function is decreasing in $A_i$. To minimize it, we should set $A_i=\frac{A_jU_j}{U_i}$. The resulting function is increasing in $A_j$ if we fix all $B_k$, so we should set $A_j=\frac{U_1}{U_j}$. We obtain the same ratio as in Case 1.

Hence in either case, we are left with minimizing the function
$$h(B_1,\dots,B_n):=\frac{1}{\sum_{j=1}^n\frac{1}{\sum_{i=1}^n\frac{1-(1-Rs_j)B_j+(1-R)s_j}{1-(1-Rs_i)B_j+(1-R)s_i}}}.$$
In other words, the reciprocal of this function, which we want to maximize, is
\begin{align*}
h_0(B_1,\dots,B_n)&:=\sum_{j=1}^n\frac{1}{\sum_{i=1}^n\frac{1-(1-Rs_j)B_j+(1-R)s_j}{1-(1-Rs_i)B_j+(1-R)s_i}}\\
&=\sum_{j=1}^n\frac{1}{\sum_{i=1}^n\frac{1-B_j+s_j(1-R+RB_j)}{1-B_j+s_i(1-R+RB_j)}}.
\end{align*}

Assume without loss of generality that $s_1\geq s_2\geq\ldots\geq s_n$. For $1\leq j\leq n$, write $h_j=\frac{1}{\sum_{i=1}^n\frac{1-B_j+s_j(1-R+RB_j)}{1-B_j+s_i(1-R+RB_j)}}$, i.e., the $j$th term of the sum that constitutes $h_0$. The last $n+1-j$ terms of the sum in the denominator of $h_j$ are at least 1 since $s_j\geq s_i$ for $i\geq j$. This implies that $h_j\leq\frac{1}{n+1-j}$, and therefore 
\begin{align*}
h_0(B_1,\dots,B_n)&=h_1+h_2+\dots+h_n\\
&\leq \frac{1}{n}+\frac{1}{n-1}+\dots+1=H_n.
\end{align*}
Equivalently, $h(B_1,\dots,B_n)\geq\frac{1}{H_n}$.

Next, since $M=\max_{i,j}\frac{S_i}{S_j}$, we have $\frac{s_i}{s_j}\leq M$ for all $i,j$. Note that
$$\frac{1-(1-Rs_j)B_j+(1-R)s_j}{1-(1-Rs_i)B_j+(1-R)s_i}\geq \frac{1}{M}$$
for all $i,j$, since both the numerator and the denominator are positive, and when seen as a function of $B_j$, the expression is monotonic and takes on values at least $\frac{1}{M}$ at $B_j=0$ and $B_j=1$. This implies that
$$h_j\leq\frac{1}{n+1-j+\frac{j-1}{M}}$$
for all $1\leq j\leq n$, and therefore
\begin{align*}
h_0(B_1,\dots,B_n)&=\sum_{j=1}^n\frac{1}{n-(j-1)\left(1-\frac{1}{M}\right)}\\
&=\sum_{j=1}^n\frac{1}{n}\cdot\frac{1}{1-\frac{j-1}{n}\left(1-\frac{1}{M}\right)}.
\end{align*}
The last sum is the left Riemann sum of the function $\frac{1}{1-x\left(1-\frac{1}{M}\right)}$. This function is increasing in $x$, so its integral between 0 and 1 is an upper bound for our sum. Hence we have
$$h_0(B_1,\dots,B_n)\leq\int_0^1\frac{1}{1-x\left(1-\frac{1}{M}\right)}dx=\frac{M\ln M}{M-1}.$$
Equivalently, $h(B_1,\dots,B_n)\geq\frac{M-1}{M\ln M}$.

Combining the two bounds, we have 
$$h(B_1,\dots,B_n)\geq\max\left(\frac{1}{H_n},\frac{M-1}{M\ln M}\right),$$
as desired.

Finally, we can obtain analogous results for revenue by essentially repeating the same argument but instead writing $A_j=\frac{(1-F(p_j))p_j}{(1-F(p_1))p_1}$ for $1\leq j\leq n$, and for any convex combination of welfare and revenue by writing $A_j$ as the appropriate convex combination of the two corresponding terms.

Done with the proof of the theorem, we now address the tightness of the approximation ratio. The upper bound $H_n$ for $h_0$ is the best possible in the sense that there exist values $B_1,\dots,B_n,s_1,\dots,s_n,R$ such that $h_0$ gets arbitrarily close to $H_n$. In particular, take $s_j=c^{-2j}$ and $B_j=1-c^{-2j+1}$ for some large constant $c$, and $R=1$. We have
\begin{align*}
h_0(B_1,\dots,B_n)&=\sum_{j=1}^n\frac{1}{\sum_{i=1}^n\frac{1-B_j+B_js_j}{1-B_j+B_js_i}}\\
&=\sum_{j=1}^n\frac{1}{\sum_{i=1}^n\frac{c^{-2j+1}+c^{-2j}-c^{-4j+1}}{c^{-2j+1}+c^{-2i}-c^{-2j-2i+1}}}\\
&=\sum_{j=1}^n\frac{1}{\sum_{i=1}^n\frac{c+1-c^{-2j+1}}{c+c^{2(j-i)}-c^{-2i+1}}}.
\end{align*}
Taking $c\rightarrow\infty$, we find that for $j\leq i$, the fraction $\frac{c+1-c^{-2j+1}}{c+c^{2(j-i)}-c^{-2i+1}}$ converges to 1, while for $j>i$, it converges to 0. Hence $h_j\rightarrow\frac{1}{n+1-j}$, and consequently $h_0\rightarrow H_n$. 

While this argument does not directly imply the tightness of the approximation ratio, we see it as strong evidence for that claim.

	\end{document}